\def\be{\begin{equation}}
\def\ee{\end{equation}}
\def\ba{\begin{array}}
\def\ea{\end{array}}

\documentclass[prl,showpacs,twocolumn,amsmath]{revtex4}
\usepackage{amsfonts}
\usepackage{mathrsfs}
\usepackage{amssymb}
\usepackage{pifont}
\usepackage{epsfig,subfigure,dsfont,amsthm,amsbsy,mathrsfs,amscd}
\def\qed{\leavevmode\unskip\penalty9999 \hbox{}\nobreak\hfill
     \quad\hbox{\leavevmode  \hbox to.77778em{%
               \hfil\vrule   \vbox to.675em%
               {\hrule width.6em\vfil\hrule}\vrule\hfil}}
     \par\vskip3pt}

\newtheorem{theorem}{Theorem}

\input amssym.def
\begin{document}
\title{\large\bf Sharing Quantum Nonlocality in Star Network Scenarios}

\author{ Tinggui Zhang$^{1,2 \dag}$ Naihuan Jing$^{2,3}$ and Shao-Ming Fei$^{1,4,5}$}
\affiliation{ ${1}$ School of Mathematics and Statistics, Hainan Normal University, Haikou, 571158, China \\
$2$ College of Sciences, Shanghai University, Shanghai, 200444, China \\
$3$ Department of Mathematics, North Carolina State University, Raleigh, NC27695, USA \\
$4$ School of Mathematical Sciences, Capital Normal University, Beijing 100048, China \\
$5$ Max-Planck-Institute for Mathematics in the Sciences, Leipzig 04103, Germany\\
$^{\dag}$ Correspondence to tinggui333@163.com}

\bigskip
\date{2022.10.12}
\begin{abstract}
The Bell nonlocality is closely related to the foundations of
quantum physics and has significant applications to security
questions in quantum key distributions. In recent years, the sharing
ability of the Bell nonlocality has been extensively studied. The
nonlocality of quantum network states is more complex. We first
discuss the sharing ability of the simplest bilocality under
unilateral or bilateral POVM measurements, and show that the
nonlocality sharing ability of network quantum states under
unilateral measurements is similar to the Bell nonlocality sharing
ability, but different under bilateral measurements. For the star
network scenarios, we present for the first time comprehensive
results on the nonlocality sharing properties of quantum network
states, for which the quantum nonlocality of the network quantum
states has a stronger sharing ability than the Bell nonlocality.
\end{abstract}

 \maketitle

\section{Introduction}
Raised by Einstein et al. \cite{abnr} in discussions on the
incompleteness of quantum mechanics in 1935, quantum Bell
nonlocality \cite{hczg} has been extensively studied and played
significant roles in a variety of quantum tasks \cite{ndsv}. Recent
studies have shown that quantum nonlocality is also a unique quantum
resource for some device-independent quantum tasks such as the key
distribution \cite{aanb}, random expansion \cite{spaa}, random
amplification \cite{rcol} and the related experiments
\cite{mxws,wlml,lkyz}.

In 2015 the authors in Ref. \cite{rsng} studied the fundamental
limits on the shareability of the Bell nonlocality with many
independent observers \cite{jsbe,chsh}, by asking whether a single
pair of entangled qubits could generate a long sequence of nonlocal
correlations under sequential measurements on one of the qubits.
This phenomenon is known as the sharing ability of quantum
nonlocality. It means that given an initial bipartite quantum state
with nonlocality, if the nonlocality can be kept under how many
unilateral or bilateral measurements. It is generally believed that
the more measurements are allowed, the stronger the nonlocal sharing
ability. Since then a series of theoretical
\cite{mnbx,assd,dass,sdas,ctdh,akak,sdsd,knab,pjbr,ztfs,ztlh,smak}
and experimental \cite{mzxc,mlmg,tcym} results on Bell's nonlocality
sharing have been derived. In \cite{assd} with equal sharpness
two-outcome measurements the authors show that at most two Bobs can
share the Bell nonlocality of a maximally entangled state with a
single Alice. It has been shown that at most two Bobs can share the
nonlocality with a single Alice by using the local realist
inequalities with three and four dichotomic measurements per
observer \cite{dass}. More recently, by an elegant measurement
strategy, the authors in \cite{pjbr} show that, contrary to the
previous expectations \cite{rsng,mnbx}, there is no limit on the
number of independent Bobs who can have an expected violation of the
CHSH inequality with one Alice. A class of initial two-qubit states,
including all pure two-qubit entangled states, that are capable of
achieving an unlimited number of CHSH inequality violations has been
presented. This fact has recently been illustrated for higher
dimensional bipartite pure states \cite{ztfs}. However, most of the
studies mentioned above are focused on the case of one-sided
(unilateral) sequential measurements (See FIG. 1). Recently, Cheng
et al. \cite{scll,cslb} explored the Bell nonlocality sharing in
bilateral sequential measurements (see FIG. 2), in which a pair of
entangled states is distributed to multiple Alices and Bobs. It is
shown that when the observers A$_1$ and B$_1$ each select their
Positive-Operator-Valued-Measure (POVM) with equal probabilities,
the nonlocality sharing between Alice$_1$-Bob$_1$ and
Alice$_2$-Bob$_2$ is impossible.
\begin{figure}[ptb]
\includegraphics[width=0.4\textwidth]{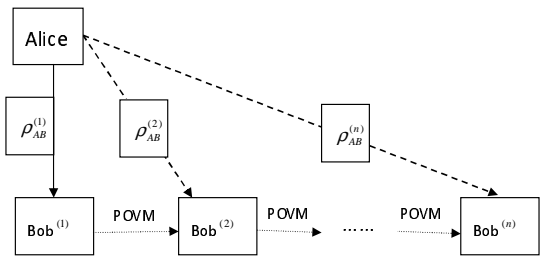}\caption{Shareability of
Bell nonlocality under unilateral measurements. Without losing
generality, in this schematic diagram, we choose Bob to make a
series of measurements, after $n-1$ measurements, we finally observe
whether there is quantum nonlocality of $\rho_{AB}^{n}$ (the quantum
state composed of Bob$^{{n}}$ and Alice ).} \label{Fig 1}
\end{figure}

\begin{figure}[ptb]
\includegraphics[width=0.4\textwidth]{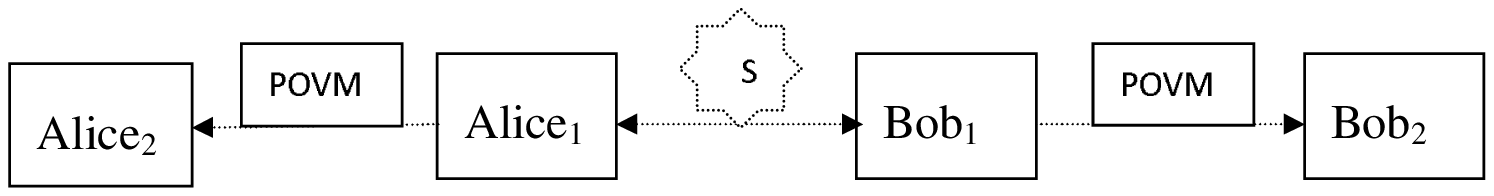}\caption{Shareability of
Bell nonlocality under bilateral measurement. Alice$_1$ and Bob$_1$
share an initial quantum state with Bell nonlocality via a quantum
resource distributor S. They both perform measurements on their own
side and send the particles to Alice$_2$ and Bob$_2$, respectively.
One identifies the nonlocality sharing between Alice$_2$ and Bob$_2$
and so on.} \label{Fig 2}
\end{figure}

With the development of quantum technology, quantum systems of
medium scale have attracted much attention
\cite{faka,hsz1,hsz2,mgon,ywuu}, such as quantum computation with
noisy intermediate-scale quantum processors \cite{nisq}. The
researches on such quantum systems have spawned another kind of
nonlocality that may be stronger than the Bell nonlocality -- the
quantum network nonlocality
\cite{cnsp,aaci,cdns,apda,mxlo,mesn,aamm,pcji,lych,lxjh,anat,mosb}.
In Ref. \cite{pcji} it was proved that any connected network
consisting of entangled pure states can exhibit genuine many-body
quantum Bell nonlocality. In Ref. \cite{apda} the authors presented
an inequality to certify the nonlocality of a star-shaped quantum
network. In Ref. \cite{lxjh} the problem of nonlocal correlation in
tree tensor networks has been studied in detail. It was shown in
Ref. \cite{mosb} that in a large class of networks with no inputs,
suitably chosen quantum color matching strategies can lead to
non-local correlations that cannot be produced in classical ways.

Recently, the nonlocality sharing problem has been mainly studied for
the Bell nonlocality. It would be also of significance to investigate the
nonlocality sharing for quantum networks. In
Ref. \cite{whxl} the authors investigated network nonlocality
sharing in the extended bilocal scenario via bilateral weak
measurements. Interestingly, when the both states $\rho_{AB}$ and
$\rho_{BC}$ are two-qubit maximally entangled pure state,
$|\psi\rangle\langle\psi|$, where
$|\psi\rangle=\frac{1}{\sqrt{2}}(|00\rangle+|11\rangle)$, by
bilateral weak measurements the network nonlocality sharing can be
revealed from the multiple violation of the Branciard-Rosset-Gisin-
Pironio (BRGP) inequalities \cite{cdns} of any
Alice$_2$-Bob-Charlie$_2$, which has no counterpart in the case of
Bell nonlocality sharing scenario.

In this paper, first we study the nonlocality sharing ability of
 bilocality quantum networks under unilateral and bilateral average
measurements. Moreover, based on bilocality methodologies, we
investigate comprehensively the problem of nonlocality sharing
ability of star quantum networks under unilateral and bilateral
average measurements. Our results incorporate the results of Ref.
\cite{whxl} and greatly generalize the range of quantum network
states.

\section{One basic fact---Nonlocal of bilocality scenario}
We first recall the simplest quantum network --- bilocality
scenario(see FIG. 3) which is given by three observers Alice, Bob
and Charlie, two sources $S_1$ and $S_2$, each source sends a
bipaitite quantum state\cite{cnsp,nqam}. Consider that Alice
receives measurement setting (or input) $x$, while Bob gets input
$y$, and Charlie $z$. Upon receiving their inputs, each party
provides a measurement result (an output), denoted by $a$ for Alice,
$b$ for Bob and $c$ for Charlie. In this context, the observed
statistics is said to be 2-local when
\begin{eqnarray}\label{ssn1} &  p(a,b,c|x,y,z) \nonumber \\ & =  \int
d\lambda d\mu q_1(\lambda)q_2(\mu)p(a|x,\lambda) p(b|y,\lambda,\mu)
p(c|z,\mu),
\end{eqnarray} where $\lambda$ and $\mu$ are hidden variables related to sources $S_1$ and $S_2$, independent shared
random variables distributed according to the densities
$q_1(\lambda)$ and $q_2(\mu)$, respectively.
\begin{figure}[ptb]
\includegraphics[width=0.4\textwidth]{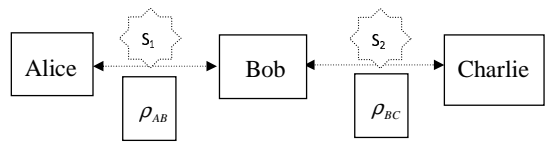}\caption{The bilocality
scenario. Resource $S_1$ distributes entangled state $\rho_{AB}$ to
Alice and Bob; resource $S_2$ distributes entangled state
$\rho_{BC}$ to Bob and Charlie. }
\label{Fig 3}%
\end{figure}

The set of 2-local correlations is non-convex. In order to
efficiently characterize the 2-local set, non-linear Bell
inequalities are required. In Refs. \cite{cnsp,cdns}, first kind
non-linear inequalities that allow one to efficiently capture
2-local correlations were derived. They are better than linear
inequalities. Consider that each party measures two possible
dichotomic observables ($a,b,c,x,y,z=0,1$). It follows that any
bilocal hidden variable (BLHV) model described by Eq.(\ref{ssn1})
must fulfill the bilocality inequality:
\begin{eqnarray}\label{ssn2}
S_{biloc}\equiv \sqrt{|I|} + \sqrt{|J|} \leq 2 ,
\end{eqnarray}
where $ I \equiv \langle(A_0 + A_1)B_0(C_0 + C_1)\rangle $, $ J \equiv
\langle(A_0 - A_1)B_1(C_0 - C_1)\rangle$, $\langle
A_xB_yC_z\rangle=\sum_{a,b,c=0,1}(-1)^{a+b+c} p(a,b,c|x,y,z)$,
$\langle O \rangle=Tr(O\rho)$ denotes the
mean value of the observable $O$ with respect to the measured state $\rho$. The violation of this inequality implies the network nonlocality of the state.

In Ref. \cite{nqam,fglr}, the authors considered the following
two-qubit state shared by Alice and Bob,
$$
\rho_{AB}=\frac{1}{4}(I_4+\vec{r}\vec{\sigma}\otimes I_2+\vec{s}I_2\otimes \vec{\sigma}+\sum_{i,j}t_{ij}^{AB}\sigma_i\otimes\sigma_j),
$$
where $\vec{\sigma} =(\sigma_x,\sigma_y, \sigma_z)$ are the standard Pauli matrices.
Here the vectors $\vec{r}$ and $\vec{s}$ represent the Bloch vectors of Alice's and
Bob's reduced states, respectively. While $t_{i,j}^{AB}$, $i, j \in {x, y,
z}$, are the entries of the correlation matrix $t^{AB}$, $I_m$ stands for the $m\times m$
identity matrix. The state $\rho_{BC}$ shared by Bob and Charlie can be expressed in a similar way. Then the maximal value of $S_{biloc}$ is shown to be
\begin{eqnarray}\label{ssn3}
S_{biloc}^{max}=2\sqrt{\delta_1\eta_1+\delta_2\eta_2},
\end{eqnarray}
where $\delta_1$ and $\delta_2$ ($\delta_1 \geq\delta_2$) are the two largest eigenvalues of the matrix $R^{AB}=\sqrt{t^{AB\dag}t^{AB}}$. Similarly, $\eta_1$ and $\eta_2$ ($\eta_1 \geq \eta_2$) are the two largest eigenvalues of the matrix $R^{BC}=\sqrt{t^{BC\dag}t^{BC}}$. Moreover, according to the Horodecki criterion \cite{rpmh}, the maximal CHSH value for $\rho_{AB}$ is given by
$S_{AB}^{max}=2\sqrt{\delta_1^2+\delta_2^2}=2\|\vec{\delta}\|$,
where $\vec{\delta}=(\delta_1,\delta_2)^{T}$ with $T$ denoting the transpose. Similarly, for
$\rho_{BC}$ one has
$S_{BC}^{max}=2\sqrt{\eta_1^2+\eta_2^2}=2\|\vec{\eta}\|$. Then it follows from Eq.
(\ref{ssn3}) that
\begin{eqnarray}\label{ssn4}
S_{biloc}^{max}=2\sqrt{\vec{\delta}\cdot\vec{\eta}}\leq 2\sqrt{\|\vec{\delta}\|\|\vec{\eta}\|} \leq \sqrt{S_{AB}^{max}S_{BC}^{max}}.
\end{eqnarray}
In the following, we will use (\ref{ssn2}), (\ref{ssn3}) and (\ref{ssn4}) to judge whether a quantum state in bilocality scenario is still nonlocal after unilateral or bilateral measurements.

\section{ Nonlocality sharing under unilateral measurement in bilocality scenario}
We first introduce the nonlocality sharing under unilateral measurement.
To begin with, Alice (Alice$_{1}$) shares an arbitrary
entangled bipartite state $\rho_{AB}$ ($\rho_{AB}^{(1)}$) with Bob.
Alice proceeds by choosing a uniformly random input, performing the
corresponding measurement and recording the outcome. Denote the
binary input and output of Alice$_k$ (Bob) by $x^{(k)}$ ($Y$) and
$A_x$ ($B$), respectively. Suppose Alice$_{1}$ performs the
measurement according to $x^{(1)}=x$ with the outcome $A_{(1)}=a$.
With equal probabilities over the inputs and outputs of Alice$_{1}$,
the post-measurement unnormalized state shared between Alice$_2$ and
Bob is given by
$$
\rho_{AB}^{(2)}=\frac{1}{2}\Sigma_{a,x}(\sqrt{A_{a|x}^{(1)}}\otimes
I_2) \rho_{AB}^{(1)}(\sqrt{A_{a|x}^{(1)}}\otimes I_2),
$$
where $A_{a|x}$ is the positive operator-valued measure (POVM) with
respect to to the outcome $a$ of Alice$^{(1)}$'s measurement for
input $x$, $I_2$ is the $2\times 2$ identity matrix. Repeating this
process, one gets the state $\rho_{AB}^{(k)}$ shared between
Alice$^{(k)}$ and Bob for $k=1,2,\cdots,n$. This process is called
unilateral measurement, see the schematic diagram shown in FIG. 4.
Our main goal is to judge whether the quantum network state composed
of $\rho_{AB}^{(n)}$ and $\rho_{BC}$ has network nonlocality.
\begin{figure}[ptb]
\includegraphics[width=0.4\textwidth]{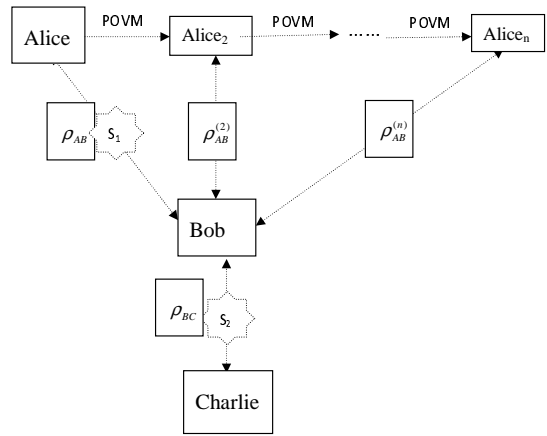}\caption{Nonlocality sharing under
unilateral measurement in bilocality scenario. Here, we choose
Alice's side for a series of measurements. In this network structure
and selective measurements, the initial state is
$\rho_{AB}\otimes\rho_{BC}$ and the final state is
$\rho_{AB}^{(n)}\otimes\rho_{BC}$ .}
\label{Fig 4}%
\end{figure}

We employ the POVMs with measurement operators $\{E,I-E\}$, where
$E=\frac{1}{2}(I_2+ \gamma \sigma_{\vec{r}})$, $\vec{r}\in R^3$ with
$\|\vec{r}\|=1$,
$\sigma_{\vec{r}}=r_1\sigma_x+r_2\sigma_y+r_3\sigma_z$, $\gamma \in
[0,1]$ is the sharpness of the measurement. For each
$k=1,2,\cdots,n$, Alice$_k$'s POVMs are given by \be\label{a00o}
A_{0|0}^{k}=\frac{1}{2}(I_2+(\cos\theta\sigma_z+\gamma_k\sin\theta\sigma_x)),
\ee \be\label{a01o}
A_{0|1}^{k}=\frac{1}{2}(I_2+(\cos\theta\sigma_z-\gamma_k\sin\theta\sigma_x)),
\ee for some $\theta\in (0,\frac{\pi}{4}]$, $k=1,2,\cdots,n$. Bob's
POVMs are given by \be\label{b00o}
B_{0|0}=\frac{1}{2}(I_2+\cos\theta\sigma_z), \ee \be\label{b01o}
B_{0|1}=\frac{1}{2}(I_2+\sin\theta\sigma_x). \ee

After Alice's side makes a finite number of sequential POVMs
measurements, we get the state $\rho_{AB}^{(n)}$ and the
following conclusion.

\begin{theorem}
If $\rho_{AB}$ is an arbitrary entangled two-qubit pure state and
$\rho_{BC}$ is the maximally entangled state
$|\psi\rangle=\frac{1}{\sqrt{2}}(|00\rangle+|11\rangle)$, then for
each $n\in \mathbb{N}$, the quantum network state composed of
$\rho_{AB}^{(n)}$ and $\rho_{BC}$ has network nonlocality, where
$\rho_{AB}^{(n)}$ stands for the state between Alice$_n$ (after
$n-1$ consecutive measurement by Alice) and Bob.
\end{theorem}

\begin{proof} In order to prove that the quantum network state composed of $\rho_{AB}^{(n)}$ and $\rho_{BC}$ has network nonlocality,
it is only necessary to prove that they violate the bilocality
inequality Eq. (\ref{ssn2}). With respect to the Alice$_k$'s POVMs,
let us define the expectation operators $A_x^{k}=A_{0|x}^{k}-A_{1|x}^{k}$ and
$B_y=B_{0|y}-B_{1|y}$ for $x,y=0,1$. Following the idea of the proof
in Ref. \cite{pjbr}, it is easy to infer that there is a strong
recursive relationship between $\rho_{AB}^{(n)}$ and $\rho_{AB}$.
The CHSH value associated with state shared between Alice$^{(n)}$ and Bob
can be similarly written as
$$I_{CHSH}^{n}=2^{2-n}(\gamma_n\sqrt{\delta_2}\sin\theta
+\sqrt{\delta_1}\cos\theta\prod_{j=1}^{n-1}(1+\sqrt{1-\gamma_j^2})).
$$
For any pure state, its corresponding correlation matrix
in Bloch representation has the maximum eigenvalue $\delta_1=1$.
According to the conclusion of Ref. \cite{pjbr}, the Bell nonlocality can be shared
under unilateral measurements. Then $\rho_{AB}^{(n)}$
would violate the CHSH inequality. Hence, according to the Horodecki
criterion \cite{rpmh},
$S_{AB}^{(n)max}=2\sqrt{(\delta_1^{(n)})^2+(\delta_2^{(n)})^2} > 2$.

On the other hand, the corresponding $\eta_1$ and $\eta_2$ for the maximally entangled pure state $\rho_{BC}$ are both $1$. Therefore, Eq. (\ref{ssn3}) can be written as
$S_{biloc}^{(n)max}=2\sqrt{\delta_1^{(n)}+\delta_2^{(n)}}$. Since
$|\delta_1^{(n)}|\leq 1$ and $|\delta_2^{(n)}|\leq 1$, we
have $\sqrt{\delta_1^{(n)}+\delta_2^{(n)}} \geq
\sqrt{(\delta_1^{(n)})^2+(\delta_2^{(n)})^2}$. That is,
$S_{biloc}^{(n)max} > 2$, which completes the proof.
\end{proof}

We note that for an entangled two-qubit mixed state $\rho_{AB}$ with
$\delta_1=1$ and $\delta_2 > 0$, according to the conclusion of Ref.
\cite{pjbr}, the Bell nonlocality of
$\rho_{AB}$ also can be shared under unilateral measurements. Therefore, similar to the proof of Theorem 1, we get
\begin{theorem}
If $\rho_{AB}$ is an entangled two-qubit mixed state with
$\delta_1=1$ and $\delta_2 > 0$, and $\rho_{BC}$ is the maximally
entangled state
$|\psi\rangle=\frac{1}{\sqrt{2}}(|00\rangle+|11\rangle)$, then for
each $n\in \mathbb{N}$, the quantum network state composed of
$\rho_{AB}^{(n)}$ and $\rho_{BC}$ has network nonlocality, where
$\rho_{AB}^{(n)}$ stands for the state between Alice$_n$ (after
$n-1$ consecutive measurement by Alice) and Bob.
\end{theorem}

\section{Nonlocal sharing under bilateral measurement of bilocality Scenario}

Furthermore, we also allow Charlie to be able to measure his party.
To begin with, Bob and Charlie (Charlie$_{1}$) share an arbitrary
entangled bipartite state $\rho_{BC}$ ($\rho_{BC}^{(1)}$). Charlie
proceeds by choosing a uniformly random input, performing the
corresponding measurement and recording the outcome. Denote the
binary input and output of Charlie$_l$ (Bob) by $z^{(l)}$ ($Y$) and
$C_z$ ($B$), respectively. Suppose Charlie$_{1}$ performs the
measurement according to $z^{(1)}=z$ with the outcome $C_{(1)}=c$.
Averaged over the inputs and outputs of Charlie$_{1}$, the
post-measurement unnormalized state shared between Bob and
Charlie$_2$ is given by
$$
\rho_{BC}^{(2)}=\frac{1}{2}\Sigma_{c,z}(\sqrt{I_2\otimes
C_{c|z}^{(1)}}) \rho_{BC}^{(1)}(I_2\otimes \sqrt{C_{c|z}^{(1)}}),
$$
where take $C_{c|z}$ is of the form of $A_{a|x}$. Repeating this process, one
gets the state $\rho_{BC}^{(l)}$ shared between Bob and
Charlie$^{(l)}$.

Let \be\label{a00o}
C_{0|0}=\frac{1}{2}(I_2+(\cos\theta\sigma_z+\gamma_l\sin\theta\sigma_x)),
\ee \be\label{a01o}
C_{0|1}=\frac{1}{2}(I_2+(\cos\theta\sigma_z-\gamma_l\sin\theta\sigma_x))
\ee for some $\theta\in (0,\frac{\pi}{4}]$, $l=1,2,\cdots,m$.
Defining the expectation operators $C_z=C_{0|z}-C_{1|z}$ and
$B_y=B_{0|y}-B_{1|y}$ for $z,y=0,1$, we have the CHSH value of the
state shared between Bob and Charlie$^{(l)}$,
$$
I_{CHSH}^{l}=2^{2-l}(\gamma_l\sqrt{\eta_2}\sin\theta
+\sqrt{\eta_1}\cos\theta\prod_{j=1}^{l-1}(1+\sqrt{1-\gamma_j^2})).
$$

Concerning the question whether the quantum network state composed
$\rho_{AB}^{(n)}$ and $\rho_{BC}^{(m)}$ (see FIG. 5) has quantum
network nonlocality, we have the following conclusion.
\begin{figure}[ptb]
\includegraphics[width=0.4\textwidth]{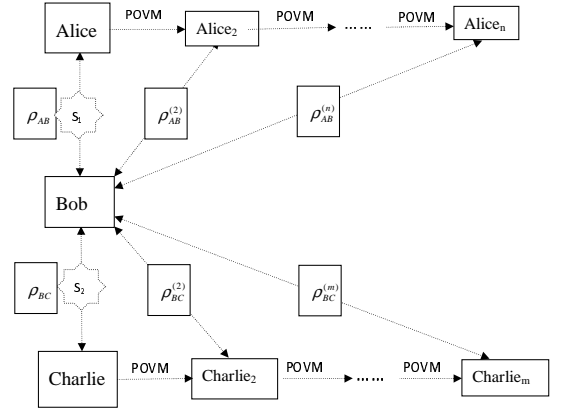}\caption{Nonlocal sharing under bilateral measurement of bilocality
Scenario. Here, we choose Alice and Charlie for a series of
measurements. In this case of network topology and selective
measurements, the initial state is $\rho_{AB}\otimes\rho_{BC}$ and
the final state is $\rho_{AB}^{(n)}\otimes\rho_{BC}^{(m)}$.}
\label{Fig 5}%
\end{figure}

\begin{theorem}
If $\rho_{AB}=\rho_{BC}$ is an entangled two-qubit pure state, then
for each $n\in \mathbb{N}$, the quantum network state composed of
$\rho_{AB}^{(n)}$ and $\rho_{BC}^{(n)}$ has network nonlocality,
where $\rho_{AB}^{(n)}$ stands for the state between Alice$_n$
(after $n-1$ consecutive measurement by Alice) and Bob,
$\rho_{BC}^{(n)}$ stands for the state between Charlie$_n$ (after
$n-1$ consecutive measurement by Charlie) and Bob.
\end{theorem}

\begin{proof} To prove the theorem one needs to prove that the bilocality inequality Eq. (\ref{ssn2}) is violated. Since the Bell nonlocality can be shared
for the quantum state $\rho_{AB}$ under unilateral measurements, $\rho_{AB}^{(n)}$
would violate the CHSH inequality. According to the Horodecki
criterion \cite{rpmh},
$S_{AB}^{(n)max}=2\sqrt{(\delta_1^{(n)})^2+(\delta_2^{(n)})^2} > 2$.

On the other hand, similar to $\rho_{AB}$, for the state $\rho_{BC}$
we also have
$S_{BC}^{(m)max}=2\sqrt{(\delta_1^{(m)})^2+(\delta_2^{(m)})^2} > 2$
for any $m\in\mathbb{N}$. Therefore, Eq. (\ref{ssn3}) now can be
written as
$S_{biloc}^{(n,m)max}=2\sqrt{\delta_1^{(n)}\delta_1^{(m)}+\delta_2^{(n)}\delta_1^{(m)}}$.

When $n=m$,
$S_{biloc}^{(n,n)max}=2\sqrt{(\delta_1^{(n)})^2+(\delta_2^{(n)})^2}
> 2$. That is, $S_{biloc}^{(n,n)max} > 2$, which completes the proof.
\end{proof}

Remark 1: When $n=m$, we can easily get that the quantum network
nonlocality is sharable. This phenomenon does not exist in Bell's
nonlocality sharing \cite{scll}. If $n \neq m$, the inequality
$S_{biloc}^{(n,m)max} > 2$ implies that the quantity
$S_{biloc}^{max}$ in (\ref{ssn3}) is greater than $2$. It should be
noted that generally the equality cannot be attained in the
inequality (\ref{ssn4}). Therefore, it is difficult to judge whether
Alice$_n$-Bob-Charlie$_m$ can share network non-locality from the
violation of the inequality (\ref{ssn2}).

Remark 2: In Ref. \cite{whxl} the sources
$S_1$ and $S_2$ send pairs of particles in the maximally entangled
state, $\rho_{AB}=\rho_{BC} = |\psi\rangle\langle\psi|$. Then the
network nonlocality sharing between Alice$_1$-Bob-Charlie$_1$ and
Alice$_2$-Bob-Charlie$_2$ can be observed. Here, by incorporating
this conclusion into Theorem 3, we see that the state does not need to be maximally
entangled, and the network nonlocality sharing between
Alice$_n$-Bob-Charlie$_n$ can be obtained for any $n\in\mathbb{N}$.

Similarly, we can easily come to the following conclusion.
\begin{theorem}
If $\rho_{AB}=\rho_{BC}$ is any entangled two-qubit mixed state with
$\delta_1=1$ and $\delta_2 > 0$, then for each $n \in \mathbb{N}$ the
quantum network state composed of $\rho_{AB}^{(n)}$ and
$\rho_{BC}^{(n)}$ has network nonlocality.
\end{theorem}

\section{Nonlocality sharing in the star network scenario}
The $n$-partite star network \cite{apda} is the natural extension of
the bilocality scenario, which is composed of n sources sharing a
quantum state between one of the $n$ nodes $A_1, A_2, \cdots, A_n$
and a central node Bob (see FIG. 6). The bilocality scenario
corresponds to the particular case of $n=2$.
\begin{figure}[ptb]
\includegraphics[width=0.4\textwidth]{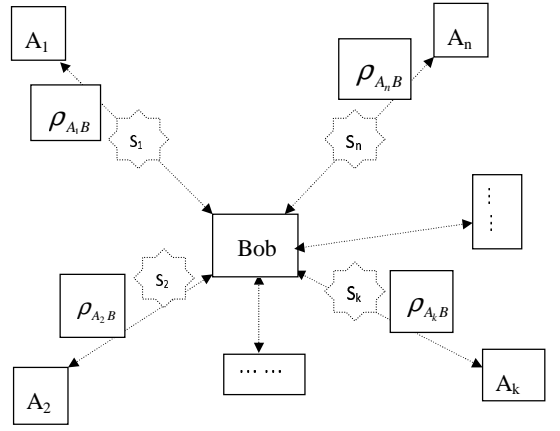}\caption{The star network scenario. Here, we choose Bob as the intermediate node,
and the source $S_i$ distributes the quantum state $\rho_{A_iB}$ to
$A_i$ and Bob, $i=1,2,\cdots,n$.}
\label{Fig 6}%
\end{figure}

The classical description of the correlations in this scenario is characterized by the
probability decomposition,
\begin{eqnarray}\label{ssn5}
& p(\{a_i\}_{i=1}^n,b|\{x_i\}_{i=1}^n,y) \nonumber \\ &=\int(\prod_{i=1}^nd\lambda_ip(\lambda_i)p(a_i|x_i,\lambda_i))p(b|y,\{\lambda_i\}_{i=1}^n).
\end{eqnarray}
As shown in [19] the following $n$-locality inequality holds,
\be\label{nns5}
\mathcal {N}_{star}=|I|^{1/n}+|J|^{1/n} \leq 1,
\ee
where
$$
I=\frac{1}{2^n}\sum_{x_1,\cdots,x_n}\langle
A_{x_1}^1A_{x_2}^2\cdots A_{x_n}^nB_0\rangle,
$$
$$
J=\frac{1}{2^n}\sum_{x_1,\cdots,x_n}(-1)^{\sum_ix_i}\langle
A_{x_1}^1A_{x_2}^2\cdots A_{x_n}^nB_1\rangle,
$$
\begin{eqnarray}
&\langle A_{x_1}^1A_{x_2}^2\cdots A_{x_n}^nB_y\rangle \nonumber\\
&=\sum_{a_1,\cdots,a_n,b}(-1)^{b+\sum_ia_i}p(\{a_i\}_{i=1}^n,b|\{x_i\}_{i=1}^n,y).\nonumber
\end{eqnarray}

According to Ref. \cite{fglr}, with respect to the generic quantum
state $\rho_{A_1B}\otimes \rho_{A_2B}\otimes\cdots\otimes
\rho_{A_nB}$, the maximal value of $\mathcal {N}_{star}$ is given by
\be\label{nns6} \mathcal {N}_{star}=\sqrt{(\prod_{i=1}^n
t_1^{A_i})^{1/n}+(\prod_{i=1}^n t_2^{A_i})^{1/n}}, \ee where
$t_1^{A_i}$ and $t_2^{A_i}$ are the two largest (positive)
eigenvalues of the matrix $t^{A_iB\dag}t^{A_iB}$ with $t_1^{A_i}\geq
t_2^{A_i}$.

\section{ Nonlocality sharing under unilateral measurements
in star network quantum states} The schematic diagram of this
situation is shown in FIG. 7, where Alice makes sequential POVM
measurements as described in Section A in the bilocality scenario,
and one gets the sequential quantum states $\rho_{AB}^{(k)}$,
$k=1,2,\cdots,m$. We need to verify the network nonlocality of the
quantum network state $\rho\equiv \rho_{AB}^{(m)}\otimes
\rho_{A_2B}\otimes \cdots \otimes\rho_{A_nB}$.
\begin{figure}[ptb]
\includegraphics[width=0.4\textwidth]{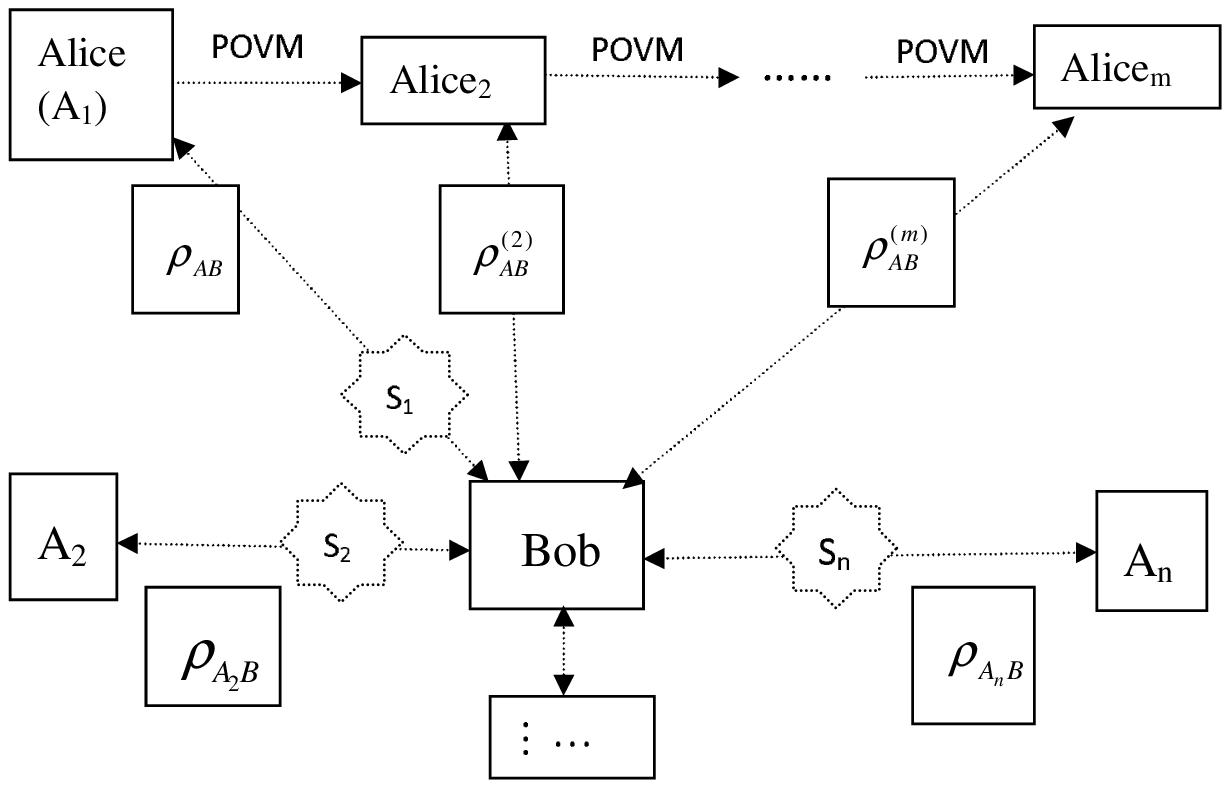}\caption{Nonlocal sharing under
unilateral measurements of star network quantum states. Here, we
choose Alice's side for a series of measurements. In this network
structure and selective measurements, the initial state is
$\rho_{A_1B}\otimes \rho_{A_2B}\otimes \cdots \otimes\rho_{A_nB}$
and the final state is $\rho_{A_1B}^{(m)}\otimes \rho_{A_2B}\otimes
\cdots \otimes\rho_{A_nB}$.}
\label{Fig 7}%
\end{figure}

\begin{theorem}
If $\rho_{A_1B}$ is an entangled pure state or an entangled mixed
state with $\delta_1=1$ and $\delta_2> 0$, and  $\rho_{A_2B}= \cdots =\rho_{A_nB}=|\psi\rangle\langle\psi|$,
then for each $n, m \in \mathbb{N}$ the quantum network state $\rho\equiv \rho_{AB}^{(m)}\otimes \rho_{A_2B}\otimes \cdots
\otimes\rho_{A_nB}$ has network nonlocality.
\end{theorem}

\begin{proof} If $\rho_{A_1B}$ is an entangled pure state or an entangled mixed
state with $\delta_1=1$ and $\delta_2> 0$, by the bi-locality case we
have that $\rho_{AB}^{(m)}$ violates the CHSH inequality. Then
according to the Horodecki criterion, we have
$S_{AB}^{(m)max}=2\sqrt{(\delta_1^{(m)})^2+(\delta_2^{(m)})^2} > 2$.
On the other hand, $\rho_{A_2B}, \cdots ,\rho_{A_nB}$ are all
maximally entangled pure states. Therefore, for the quantum network
state $\rho= \rho_{AB}^{(m)}\otimes \rho_{A_2B}\otimes \cdots
\otimes\rho_{A_nB}$, Eq. (\ref{nns6}) can be expressed as $\mathcal
{N}_{star}=\sqrt{(\delta_1^{(m)})^2+(\delta_2^{(m)})^2} =
\frac{S_{AB}^{(m)max}}{2} > 1$. Moreover, the inequality (\ref{nns5}) is
violated, and the quantum network state $\rho$ has network nonlocality.
\end{proof}

\section{ Nonlocality sharing under multilateral measurements in star network quantum states}
A schematic diagram of this situation is shown in FIG. 8. Alice and
Charlie make sequential POVM measurements as described in Section B
in bilocality scenario. We get the sequential quantum states
$\rho_{A_1B}^{(k)}$ and $\rho_{A_2B}^{(l)}$, $k=1,2,\cdots,m$,
$l=1,2,\cdots,t$ and so on. Our problem is to identify the network
nonlocality of the quantum network state
$$\rho\equiv \rho_{A_1B}^{(m)}\otimes \rho_{A_2B}^{(l)}\otimes \cdots
\otimes \rho_{A_pB}^{(s)} \otimes \rho_{A_{p+1}B}\otimes\cdots
\otimes\rho_{A_nB}.$$
\begin{figure}[ptb]
\includegraphics[width=0.4\textwidth]{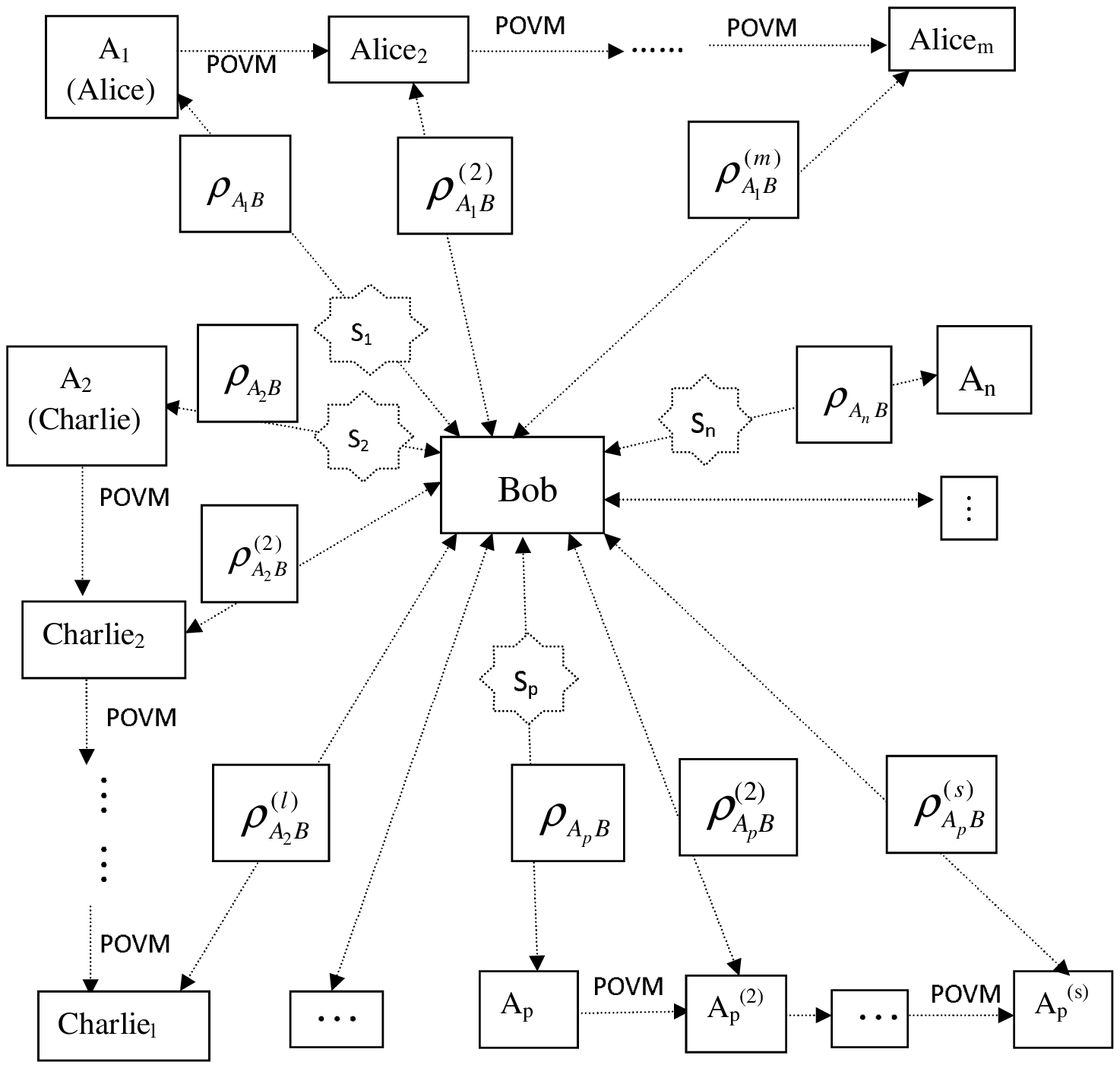}\caption{Nonlocality sharing under
multilateral measurements of star network quantum states. Here, we
choose $A_1(Alice),A_2(Charlie),\cdots, A_p$ for a series of
measurements. In this network structure and selective measurements,
the initial state is $\rho_{A_1B}\otimes \rho_{A_2B}\otimes \cdots
\otimes\rho_{A_nB}$ and the final state is $\rho_{A_1B}^{(m)}\otimes
\rho_{A_2B}^{(l)}\otimes \cdots
\otimes\rho_{A_pB}^{(s)}\otimes\cdots \otimes \rho_{A_nB}$.}
\label{Fig 8}%
\end{figure}

From the analysis on the Bilocality situation above and the
inequality (\ref{nns5}) and (\ref{nns6}) of the
nonlocality for the star quantum network, it is not difficult
to obtain the following conclusion.
\begin{theorem}
If $\rho_{A_1B}=\rho_{A_2B}=\cdots=\rho_{A_pB}$ is an entangled pure
state or an entangled mixed state with $\delta_1=1$ and $\delta_2
>0$, for $p=2,3,\cdots,n$, and $\rho_{A_{p+1}B}= \cdots
=\rho_{A_nB}=|\psi\rangle\langle\psi|$, then for each $ m,l,\cdots,s
\in \mathbb{N}$ the quantum network state $\rho\equiv
\rho_{A_1B}^{(m)}\otimes \rho_{A_2B}^{(l)}\otimes \cdots \otimes
\rho_{A_pB}^{(s)} \otimes \rho_{A_{p+1}B}\otimes\cdots
\otimes\rho_{A_nB}$ has network nonlocality.
\end{theorem}

\section{Conclusions and discussions}
Quantum nonlocality is a distinctive feature of quantum mechanics.
The nonlocal characteristics of quantum networks are more
complicated due to their non-convexity, nonlinearity, etc. The
research on the nonlocality sharing ability of quantum networks has
important theoretical significance for developments of such as
quantum repeaters. Since nonlocality is an important quantum
resource, starting from a nonlocally correlated quantum state, it
would be quite desirable to be able to maintain the nonlocality
under sequential measurements. In the bilocality case, we have shown
that under unilateral measurements, the nonlocality can be shared
under any times of sequential measurements. With the same
measurements and the same times of measurements on both sides of
Alices and Charlies, it is possible for arbitrarily many Alices and
Charlies to share the locality with a single Bob by using a pure or
mixed entangled state. We have also investigated the shareability of
star network nonlocality. It has been shown that from our
measurement schemes the nonlocality of the star network quantum
states can be shared under unilateral or multilateral measurements.

Our results may also highlight researches on sharing general
multipartite quantum nonlocalities \cite{cxwt} and other quantum
correlations such as quantum steerability \cite{sdsa}, entanglement
\cite{asau,adaa} and coherence \cite{sasm,hwfh}. Our approach may
also suggest the related applications in randomness generation
\cite{fmrm}, quantum teleportation \cite{sasa} and random access
codes \cite{kanb}. It would be also interesting to explore the
nonlocality sharing ability of high-dimensional \cite{ztfs,ztlh} or
multipartite quantum network states in other network scenarios.
Moreover, we have constructed the dichotomic POVM measurement
operators in terms of the Pauli operators. As the Pauli operators
are easily implemented in experiments, the POVM operators we
constructed may have potential advantages in some specific
experimental implementations \cite{mzxc,mlmg,tcym}.

\bigskip
\section*{Acknowledgments}
This work is supported by the National Natural Science Foundation of China (NSFC) under Grant Nos. 12126314, 12126351, 11861031, 12075159 and 12171044; the Hainan Provincial Natural Science Foundation of China under Grant Nos.121RC539, the specific research fund of the Innovation Platform for Academicians of Hainan Province under Grant No. YSPTZX202215,
Beijing Natural Science Foundation (Grant No. Z190005); Academy for Multidisciplinary Studies, Capital Normal University; Shenzhen Institute for Quantum Science and Engineering, Southern University of Science and Technology (No. SIQSE202001).


\begin{thebibliography}{99}

\bibitem{abnr} A. Einstain, B. Podolsky and N. Rosen, Phys. Rev.
{\bf 47}, 777 (1935).
\bibitem{hczg} H. Cao, Z. Guo, Sci. China Phys. Mech. Astron. {\bf 62}, 30311 (2019).
\bibitem{ndsv} N. Brunner, D. Cavalcanti, S. Pironio, V. Scarani and S. Wehner, Rev. Mod. Phys. {\bf 86}, 419 (2014).
\bibitem{aanb} A. Acin, N. Brunner, N. Gisin, S. Massar, S. Pironio and V.
Scarani, Phys. Rev. Lett. {\bf 98}, 230501 (2007).

\bibitem{spaa} S. Pironio, et al., Nature, {\bf 464}, 1021 (2010).

\bibitem{rcol} R. Colbeck and R. Renner, Nature Physics, {\bf 8}, 450 (2012).

\bibitem{mxws} M. H. Li, et al., Phys. Rev. Lett. {\bf 126}, 050503 (2021).

\bibitem{wlml} W. Z. Liu, et al., Nature Physics, {\bf 17}, 448 (2021).

\bibitem{lkyz} L. K. Shalm, et al., Nature Physics,
{\bf 17}, 452 (2021).

\bibitem{rsng} R. Silva, N. Gisin, Y. Guryanova, and S. Popescu, Phys. Rev.  Lett. {\bf 114}, 250401
(2015).
\bibitem{jsbe} J. S. Bell, Physics, {\bf 1}, 195 (1964).
\bibitem{chsh} J. F. Clauser, M.A. Horne, A. Shimony and R.A. Holt,
Phys. Rev. Lett. {\bf 23}, 880 (1969).

\bibitem{mnbx} S. Mal, A. Majumdar, and D. Home, Mathematics {\bf 4}, 48 (2016).



\bibitem{assd} A. Shenoy H., S. Designolle, F. Hirsch, R. Silva, N. Gisin, and
N. Brunner,  Phys. Rev. A {\bf 99},  022317 (2019).

\bibitem{dass} D. Das, A. Ghosal, S. Sasmal, S. Mal, and A. S. Majumdar,
Phys. Rev. A {\bf 99}, 022305 (2019).

\bibitem{sdas} S. Datta and A. S. Majumdar, Phys. Rev. A {\bf 98}, 042311 (2019).
\bibitem{ctdh} C. Ren, T. Feng, D. Yao, H. Shi, J. Chen, and X. Zhou,  Phys.
Rev. A {\bf 100}, 052121(2019).


\bibitem{akak} A. Kumari and A. K. Pan, Phys. Rev. A {\bf 100}, 062130 (2019).

\bibitem{sdsd} S. Saha, D. Das, S. Sasmal, D. Sarkar, K. Mukherjee, A.  Roy,
and S. S. Bhattacharya, Quantum Inf. Process. {\bf 18}, 42 (2019).

\bibitem{knab} K. Mohan, A. Tavakoli and N. Brunner, New J. Phys. {\bf 21},
083034 (2019).

\bibitem{pjbr} P. J. Brown and R. Colbeck, Phys. Rev. Lett. {\bf 125}, 090401
(2020).
\bibitem{ztfs} T. Zhang and S. M. Fei, Phys. Rev. A {\bf 103}, 032216
(2021).
\bibitem{ztlh} T. Zhang, Q. Luo and X. Huang, Quantum Inf. Process.
(2022)

\bibitem{smak} S. Mukherjee and A. K. Pan, Phys. Rev. A {\bf 104},
062214 (2021).

\bibitem{mzxc} M. J. Hu, Z. Y. Zhou, X. M. Hu, C. F. Li, G. C. Guo,  and Y. S.
Zhang, Npj. Quantum. Inform. {\bf 4}, 63 (2018).

\bibitem{mlmg} M. Schiavon, L. Calderaro, M. Pittaluga, G. Vallone, and  P.
Villoresi, Quantum Sci. Technol. {\bf 2}, 015010 (2017).

\bibitem{tcym} T. Feng, C. Ren, Y. Tian, M. Luo, H. Shi, J. Chen, and  X. Zhou,
Phys. Rev. A {\bf 102}, 032220 (2020).

\bibitem{scll} S. Cheng, L. Liu, T. J. Baker, M. J. W. Hall, Phys. Rev. A {\bf 104}, L060201 (2021).


\bibitem{cslb} S. Cheng, L. Liu, T. J. Baker, M. J. W.
Hall, Phys. Rev. A {\bf 105}, 022411 (2022).

\bibitem{faka} F. Arute, K. Arya, et al. Nature {\bf 574}, 505 (2019).

\bibitem{hsz1} H. S. Zhong, et al. Science, {\bf 370} (6523), 1460 (2020).

\bibitem{hsz2} H. S. Zhong, et al. Phys. Rev. Lett. {\bf 127}, 180502
(2021).

\bibitem{mgon} M. Gong, et al. Science, {\bf 372}(6545), 948 (2021).

\bibitem{ywuu} Y. Wu, et al. Phys. Rev. Lett. {\bf 127}, 180501
(2021).

\bibitem{nisq} J. M. Liang, S. Q. Shen, M. Li, S. M. Fei, Quant. Inform. Processing {\bf 21}, 23 (2022).\\
J. M. Liang, S. J. Wei, S. M. Fei, Sci. China Phys. Mech. \& Astron.
{\bf 65}, 250313 (2022).

\bibitem{cnsp} C. Branciard, N. Gisin, and S. Pironio, Phys. Rev. Lett. {\bf 104}, 170401 (2010).

\bibitem{aaci} A. Acin, et al.  Nat.
Commun. {\bf 2}, 184 (2011).

\bibitem{cdns} C. Branciard, D. Rosset, N. Gisin, and
S. Pironio, Phys. Rev. A {\bf 85}, 032119 (2012).

\bibitem{apda} A. Tavakoli, P. Skrzypczyk, D. Cavalcanti, and A.
Acin, Phys. Rev. A {\bf 90}, 062109 (2014).

\bibitem{mxlo} M. X. Luo, Phys. Rev. Lett. {\bf 120}, 140402 (2018).

\bibitem{mesn} M. O. Renou, E. Brumer, S. Boreiri, N. Brunner, N. Gisin and S.
Beigi, Phys. Rev. Lett. {\bf 123}, 140401 (2019).

\bibitem{aamm} A. Tavakoli, A. Pozas-Kerstjens, M. X. Luo and  M. O. Renou, Rep. Prog. Phys. {\bf 85}, 056001 (2022).

\bibitem{pcji} P. Contreras-Tejada, C. Palazuelos and J. I. de Vicente,
Phys. Rev. Lett. {\bf 126}, 040501 (2021).

\bibitem{lych} L. Y. Hsu and C. H. Chen, Phys. Rev. Research
{\bf 3}, 023139 (2021).

\bibitem{lxjh} L. Yang, X. Qi, and J. Hou, Phys. Rev. A
{\bf 104}, 042405 (2021).

\bibitem{anat} A. Pozas-Kerstjens, N. Gisin and A. Tavakoli, Phys. Rev. Lett. {\bf 128}, 010403 (2022).

\bibitem{mosb} M. O. Renou and S. Beigi,
Phys. Rev. Lett. {\bf 128}, 060401 (2022).

\bibitem{whxl} W. Hou, X. Liu, and C. Ren, Phys. Rev. A {\bf 105},
042436 (2022).

\bibitem{nqam} N. Gisin, Q. Mei, A. Tavakoli, M.O. Renou, and N.
Brunner, Phys. Rev. A {\bf 96}, 020304(R) (2017).

\bibitem{fglr} F. Andreoli, G. Carvacho, L. Santodonato, R. Chaves and
F. Sciarrino, New J. Phys. {\bf 19}, 113020 (2017).

\bibitem{rpmh}R. Horodecki, P. Horodecki and M. Horodecki, Phys.  Lett. A {\bf 200},
340 (1995).
\bibitem{cxwt} C. Ren, X. Liu, W. Hou, T. Feng, and X. Zhou, Phys. Rev. A {\bf 105}, 052221
(2022).

\bibitem {sdsa} S. Sasmal, D. Das, S. Mal, and A. S. Majumdar,  Phys. Rev.  A
{\bf 98}, 012305 (2018).


\bibitem {asau} A. Bera, S. Mal, A. Sen De, and U. Sen, Phys. Rev.  A
{\bf 98}, 062304 (2018).


\bibitem {adaa} A. G. Maity, D. Das, A. Ghosal, A. Roy, and A. S. Majumdar, Phys. Rev. A {\bf 101}, 042340
(2020).


\bibitem {sasm} S. Datta, and A. S. Majumdar, Phys. Rev. A  {\bf
98}, 042311 (2018).
\bibitem{hwfh} M. L. Hu, J. R. Wang, and H. Fan, Sci. China Phys. Mech. Astron. {\bf 65}, 260312 (2022).
\bibitem {fmrm} F. J. Curchod, M. Johansson, R. Augusiak, M. J. Hoban,  P.
Wittek, and A. Acin,  Phys. Rev. A {\bf 95}, 020102  (2017).

\bibitem {sasa} S. Roy, A. Bera, S. Mal, A. Sen De, and U. Sen, Phys. Lett. A {\bf 392}, 127143 (2021).

\bibitem {kanb} K. Mohan, A. Tavakoli, and N. Brunner, New J.
Phys. {\bf 21}, 083034 (2019).



\end{thebibliography}
\end{document}